\documentclass[11pt]{amsart}

\usepackage{float}
\usepackage{macros}
\date{}

\begin{document}
\title{Integrable lattice models from four-dimensional field theories}
\author{Kevin Costello}
\thanks{Partially supported by NSF grant DMS 1007168, by a Sloan Fellowship, and by a Simons Fellowship in Mathematics.}

\address{Department of Mathematics, Northwestern University.}
\email{costello@math.northwestern.edu}

\maketitle
\begin{abstract}
This paper gives a general construction of an integrable lattice model (and a solution of the Yang-Baxter equation with spectral parameter) from a four-dimensional field theory which is a mixture of topological and holomorphic.  Spin-chain models arise in this way from a twisted, deformed version of $N=1$ gauge theory.  This note is  based on the longer paper arXiv:1303.2632.
\end{abstract}
\renewcommand{\L}{\mathscr{L}}

\section{Introduction}
Integrable lattice models have a long and fruitful history in physics, dating back to Heisenberg's work on the XXX model. Integrability of the XXX and related models was proved by Baxter, Bethe, Yang and others in the 60's and 70's.  A key insight of this work is that integrability follows from the fact that the vertex interaction of the model -- encoded by the $R$-matrix -- satisfies the Yang-Baxter equation. 

If we take an integrable model and perturb it a small amount, it will typically no longer be integrable.  The physical properties of the perturbed model will be essentially identical, however.  One can therefore ask the following question: where do integrable models come from? \footnote{This paragraph paraphrases some comments made by Okounkov in a lecture in 2013.} 

In this note (which is a summary of the long paper \cite{Cos13}) I propose the following answer: integrable models arise from four-dimensional field theories which are is topological in two real directions and holomorphic in one complex direction.  I show that every such field theory, equipped with some line operators in the topological directions, leads to a two-dimensional integrable lattice model.  The correspondence is as follows.
\begin{enumerate}
\item The partition function of the lattice model is equal to the expectation value of a configuration of line operators on a product of a Riemann surface $\Sigma$ and a topological two-torus $T^2$.
\item The Hilbert space of the lattice model is the Hilbert space of the field theory on a Riemann surface $\Sigma$ times a topological $S^1$, in the presence of line operators which end at points on the circle.
\item The transfer matrix is the operator on the Hilbert space associated to $\Sigma \times S^1$ arising from a line operator parallel to the $S^1$.
\item The spectral parameter is a meromorphic parameter on $\Sigma$.
\item The Boltzmann weights (or $R$-matrix) of the lattice model arises from the operator product expansion of line operators.
\end{enumerate}
Kapustin showed that any $N=2$ field theory admits a twist of this form.  I showed in \cite{Cos13} that $N=1$ pure gauge theory can be deformed and twisted to yield a theory of this form.  This deformed $N=1$ gauge theory has a Wilson operator invariant under the supercharge we use to twist. The main result of \cite{Cos13}, which I sketch here, states that the integrable lattice model associated to a twisted, deformed $N=1$ gauge theory, with gauge group $SU(n)$ and Wilson operator in a representation $V$ of $SU(n)$, is the higher spin-chain system associated to the $SU(n)$ representation $V$. (Thus, in the case that $n = 2$ and $V$ is the fundamental representation, we find the Heisenberg $XXX$ model).  

The result generalises to other semi-simple gauge groups: however, for $G \neq SU(n)$, the Wilson operator associated to a $G$-representation $V$ may have a quantum anomaly.  This anomaly occurs if $V$ can not be lifted to a representation of the Yangian $Y(\g)$ of the Lie algebra $\g$ of $G$. 

Kapustin's holomorphic/topological twist of $N=2$ gauge theories admits both Wilson and t'Hooft operators.    The construction of this paper, applied to Kapustin's twist, will yield an integrable lattice model associated to any $N=2$ theory, whose partition function is the expectation value of a configuration of Wilson and t'Hooft operators. 

There are several other known relationships between integrable lattice models and four-dimensional field theories. One was introduced by Nekrasov and Shatashvili in \cite{NekSha09}, and developed mathematically by Maulik and Okounkov in the beautiful paper \cite{MauOko12}.  It seems that these two connections between field theories and integrable systems are completely different. Indeed, Nekrasov and Shatashvili show that the spin-chain system for an ADE group $G$ is associated to the $N=2$ quiver gauge theory with ADE quiver corresponding to $G$, whereas in this paper we find that the same spin-chain system arises from the $N=1$ gauge theory with gauge group $G$.    

Another relationship between integrable systems and gauge theories was developed by Yamazaki in \cite{Yam13}.  Again, this appears to be different from the relationship developed here, in that in Yamazaki's work the Yang-Baxter equation is derived from Seiberg duality, whereas here the Yang-Baxter equation is much easier to derive. 

Of course, there is also the well-known connection between $N=4$ gauge theory and the Yangian (see e.g. \cite{Fer11}).  This as also, as far as I know, unrelated to the results of this note.

\subsection{Acknowledgements}
I'd like to thank Kolya Reshitikhin, Nick Rozenblyum, Josh Shadlen, and Edward Witten for helpful conversations. 

\section{Integrable lattice models}

In this section, I will define the concept of integrable lattice model from the vertex-model point of view (i.e. from the discrete version of the path-integral approach to quantum field theory). 

Let $V, W$ be finite-dimensional vector spaces. Let 
$$
\check{R}: V \otimes W \to W \otimes V
$$
be a linear map.  From this data, we will construct a two-dimensional discrete lattice model. 

\begin{example}
The Heisenberg XXX-model (or 6-vertex model) has $V = W = \C^2$, and $\check{R}$-matrix defined by
$$
\check{R} (v \otimes w) = w \otimes v + \tfrac{1}{z} c (w \otimes v) 
$$
where $c \in \sl_2 \otimes \sl_2$ is the quadratic Casimir, and $z \in \C^\times$ is a parameter called the spectral parameter. 

Note that some authors write the Heisenberg model in terms of 
$$R = \sigma \circ \check{R},$$
where $\sigma:  W \otimes V \to V \otimes W$ is the isomorphism which interchanges the factors.
 \end{example}

\subsection{}
Suppose that $L$ is an $n \times m$ doubly periodic lattice. Thus, $L$ is the quotient of the standard infinite square lattice in $\R^2$ with vertices $\Z \times \Z$ by the subgroup $n \Z \times m \Z \subset \Z \times \Z$, acting by translation.  

Choose a basis $e_i$ of $V$ and $f_j$ of $W$.
\begin{definition}
A \emph{configuration} of the lattice model is a way of labelling every horizontal edge of $L$ by a basis element of $W$, and every vertical element by a basis element of $V$.
\end{definition}
Suppose that we have a configuration $\sigma$ on $L$, and a vertex $v$ of $L$. Suppose that, in this configuration, the edges incident to the vertex $v$ are labelled by basis elements $e_{i_1},e_{i_2}$ of $V$ and $f_{j_1}, f_{j_2}$ of $W$. We define $\check{R}(v,\sigma)$ to be the matrix element
$$
\check{R}(v,\sigma) = \check{R}_{i_1,i_2,j_1,j_2}
$$
associated to these basis vectors. This is the Boltzmann weight of the lattice model at the vertex.   Thus, we should picture the matrix $\check{R}$ as in figure \ref{fig:interaction}.
\begin{figure}
\includegraphics[scale=1]{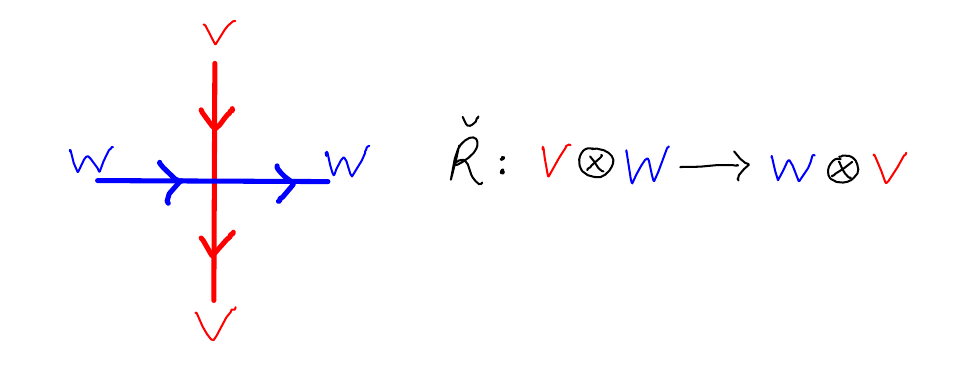}
\caption{The $\check{R}$-matrix. \label{fig:interaction}}
\end{figure}
 
\begin{definition}
The \emph{partition function} of the lattice model is defined by
$$
Z(L, \check{R}) = \sum_{\text{configurations } \sigma} \prod_{\text{vertices } v} \check{R}(v,\sigma).
$$
\end{definition}

\subsection{}
We can re-express our lattice model in the Hamiltonian formalism.
\begin{definition}
The \emph{Hilbert space} of the lattice model is $V^{\otimes n}$.
\end{definition}
\begin{definition}
The \emph{transfer matrix} $T : V^{\otimes n} \to V^{\otimes n}$ is defined as follows.  Let us view $\check{R}$ as an element of $\op{End}(V) \otimes \op{End}(W)$.  Then, $\check{R}^{\otimes n}$ is an element of $\op{End}(V)^{\otimes n} \otimes \op{End}(W)^{\otimes n}$.  We can apply the $W$-composition map
$$
\op{End}(W)^{\otimes n} \to \op{End}(W)
$$
to get an element $\check{R} \circ_W \dots \circ_W \check{R} \in \op{End}(V)^{\otimes n} \otimes \op{End}(W)$. Finally, we can take the trace over $W$ to get an element
$$
T = \op{Tr}_W \left(\check{R} \circ_W \dots \circ_W \check{R} \right) \in \op{End}(V)^{\otimes n} = \op{End}(V^{\otimes n}).
$$
\end{definition}
The transfer matrix is illustrated in figure \ref{fig:transfer}.
\begin{figure}
\includegraphics{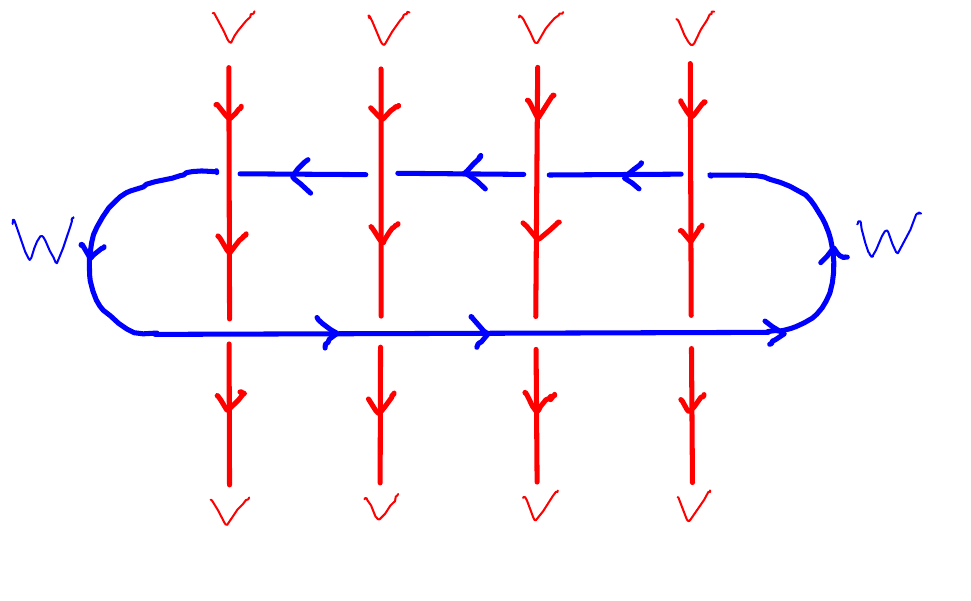}
\caption{The transfer matrix. The horizontal line has been closed into a circle because we are taking the trace of the operator in that direction. \label{fig:transfer}}
\end{figure}
\begin{lemma}
The partition function can be expressed in terms of the transfer matrix by
$$
Z(L, \check{R} ) = \op{Tr}_{V^{\otimes n}} T^m. 
$$
\end{lemma}
\begin{proof}
This is a standard lemma in the theory of vertex models.  The point is that the transfer matrix is a certain trace of compositions of the operator $\check{R}$. If we write out the right hand side of this equality explicitly using the basis we have chosen for $V$ and $W$ we find precisely the expression for the partition function we wrote down earlier.
\end{proof}
\begin{remark}
In the usual terminology, the transfer matrix is not the same as the Hamiltonian, but one can be expressed in terms of the other.  Roughly, the Hamiltonian is the log of the transfer matrix.
\end{remark}

\section{Integrability}
Next, we will say what it means for a lattice model as above to be integrable.  The basic idea is that a lattice model is integrable if there are an infinite number of operators on the Hilbert space  $V^{\otimes n}$ which commute with the transfer matrix.  We will consider a more precise form of integrability, however. 

Suppose that the matrix $\check{R}$ depends holomorphically on a complex parameter $z$ taking values in a Riemann surface $\Sigma$. The parameter $z$ is called the spectral parameter. ($\Sigma$ can be non-compact; for spin-chain systems, $\Sigma = \mbb{CP}^1 \setminus\{0\}$.) Then, the transfer matrix will also depend on $z$, so that we get a one-parameter family of matrices 
$$T(z) : V^{\otimes n} \to V^{\otimes n}.$$
\begin{definition}
The lattice model is \emph{integrable} if 
$$
[T(z), T(z')] = 0
$$
for all $z,z' \in \Sigma$. 
\end{definition}
If we fix one value $z_0$ of $z$, we see that each Taylor term of $T(z)$ expanded around $z_0$ commutes with $T(z_0)$.  In this way we find an infinite number of operators commuting with $T(z_0)$.

One can ask if there's a condition on $\check{R}$ which implies that the transfer matrices for different values of $z$ commute.  There is: this condition is called the \emph{Yang-Baxter equation}.  

To phrase this condition in the way that appears in field theory, we need to be slightly more general. Suppose that ${V}$, ${W}$ are holomorphic vector bundles on a Riemann surface $\Sigma$.  We will denote the fibres at points $p,q \in \Sigma$ by ${V}_p$ and ${W}_q$.  

Suppose that, for each $p,q \in \Sigma$, with $p \neq q$, we have an isomorphism
$$
\check{R}({V}_p,{W}_q) : {V}_p \otimes {W}_q \to {W}_q \otimes {V}_p.
$$
Suppose that this isomorphism varies holomorphically with 
$$
(p,q) \in \Sigma \times \Sigma \setminus \Diag.
$$
It makes sense to ask that $\check{R}(p,q)$ varies holomorphically, because the vector bundles ${V}$ and ${W}$ are holomorphic.

The transfer matrix 
$$T(p,q) : {V}_p^{\otimes n} \to {V}_p^{\otimes n}$$
defined using $\check{R}({V}_p, {W}_q)$ will also vary holomorphically with $p$ and $q$ (which must be distinct).  Thus, fixing $p$, $T(p,q)$ is a holomorphic map from $\Sigma \setminus \{p\}$ to $\op{End} ({V}_p^{\otimes n})$. As above, we say that the theory is integrable if, for fixed $p$, we have
$$
[T(p,q), T(p,q')] = 0.
$$

We will see that integrability will follow from some extra data related to $\check{R}({V}_p, {W}_q)$.  Suppose that we have a linear map
$$
\check{R}({W}_p, {W}_q) : {W}_p \otimes {W}_q \iso {W}_q \otimes {W}_p
$$
defined for each $p \neq q \in \Sigma$. Again, we assume that $\check{R}({W}_p, {W}_q)$ varies holomorphically with $p$ and $q$. 
\begin{figure}
\includegraphics[scale=0.8]{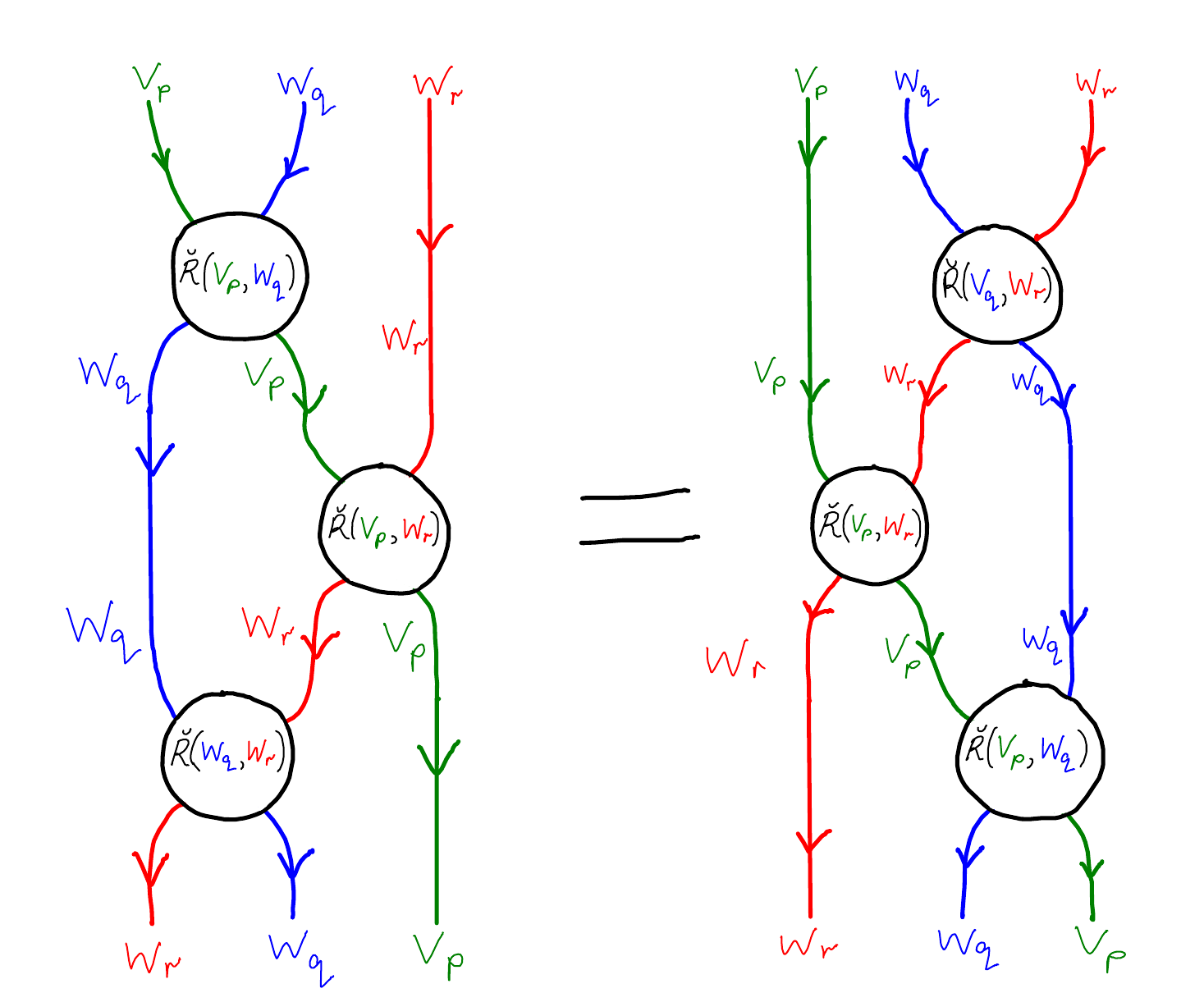}
\caption {The Yang-Baxter equation. \label{fig:YBE}}
\end{figure}

\begin{definition}
$\check{R}({V}_p, {W}_q)$ and $\check{R}({W}_p, {W}_q)$ satisfy the \emph{Yang-Baxter equation} if for every triple of distinct points $p,q,r \in \Sigma$ we have
$$
\check{R}({W}_q, {W}_r)\check{R}({V}_p, {W}_r)\check{R}({V}_p, {W}_q) = \check{R}({V}_p, {W}_q)\check{R}({V}_p, {W}_r)\check{R}({W}_q, {W}_r).
$$
\end{definition}

\begin{remark}
Normally, the Yang-Baxter equation is written in terms of $R = \sigma \circ \check{R}$, which is why the equation above might look slightly different to what is usually written in the literature. 
\end{remark}
The Yang-Baxter equation is illustrated in figure \ref{fig:YBE}. One can show quite easily that the Yang-Baxter equation, together with the statement that 
$$\check{R}({W}_p, {W}_q) \check{R}({W}_q,{W}_p) = 1$$
implies integrability.

\section{Integrable models from four-dimensional field theories}
In this section I will explain how to construct an integrable model from a four-dimensional field theory which is topological in 2 real directions and holomorphic in 1 complex direction.

The reader should be aware that what I mean by topological is a little weaker than what some authors mean. For the purposes of this paper, a two-dimensional topological field theory is a field theory which can be defined on \emph{framed} topological surfaces.  According to Lurie's \cite{Lur09} classification, such topological field theories correspond to (smooth and proper) dg categories.  This should be compared to the case of TFTs defined on oriented surfaces: according to \cite{Cos07a,Lur09} these are classified by smooth and proper Calabi-Yau categories.  

A similar remark applies to my use of the term ``holomorphic field theory''.  Only a weak version of this concept is required: we do not require that the theory is invariant under the Virasoro algebra, or that it can be defined on an arbitrary Riemann surface.  It is enough to have a theory which is only invariant under translation on $\C$, and so can be defined on Riemann surfaces equipped with a nowhere-vanishing holomorphic one-form.  

Suppose we have a field theory on $\R^4 = \C^2$. (By ``field theory'' we could mean classical field theory, specified by some Lagrangian; or quantum field theory, encoded by a factorization algebra as in \cite{CosGwi11}).  Suppose this field theory is translation-invariant.  We say a field theory is \emph{topological} if the action of the Abelian Lie algebra $\R^4$ on our theory is homotopically trivial (i.e. BRST exact).   Similarly, we say a field theory is \emph{holomorphic} if the action of the Abelian Lie algebra spanned by $\dpa{\zbar}, \dbar{\br{w}}$ in the complexified translation algebra is homotopically trivial.  We say a theory is mixed holomorphic/topological if the action of the $3$-dimensional Abelian Lie algebra spanned by $\dpa{w}, \dpa{\br{w}}, \dpa{\zbar}$ is homotopically trivial.  In this situation, the theory is topological in the $w$-plane and holomorphic in the $z$-plane.  

One can ask where holomorphic/topological field theories come from.  One way they arise is by twisting supersymmetric theories (mathematicians might consult \cite{Fre99,DelFre99,Cos11b} for a primer on supersymmetry and twisting).  

Suppose we have a supersymmetric theory on $\R^4$ in Euclidean signature.  Thus, the theory is acted on by a complex super-translation Lie algebra of the form $\pi S \oplus \C^4$, where $S$ is a complex spin representation of $\op{Spin}(4)$ and $\C^4$ is the complexification of the vector representation. The Lie bracket is given by a complex-linear symmetric and $\op{Spin}(4)$-equivariant map 
$$
\Gamma : S \otimes S \to \C^4.
$$
Suppose we have some $Q \in S$ with the property that $[Q,Q] = 0$. Then, we can consider the twisted field theory, defined by adding $Q$ to the BRST operator of the theory.  Observables of the twisted theory are the $Q$-cohomology of observables of the original theory\footnote{I'm glossing over the role of $R$-symmetry in this story, which is needed to ensure the twisted theory is $\Z$-graded.  See e.g. \cite{Cos11b} for more details on twisting. In most treatments, an important part of the twisting procedure is to change the action of $\op{Spin}(4)$ on the theory so that the supercharge $Q$ is invariant.  This step is not important for the present discussion, because we are interested in field theories which can not be made $\op{Spin}(4)$-invariant.}.

The twisted theory has the property that any translation vector which can be written as $[Q,Q']$ for some $Q' \in S$ is homotopically trivial, the homotopy being given by $Q'$. 

In particular, if $\Im Q$ is spanned by $\dpa{\zbar}, \dpa{w}, \dpa{\br{w}}$, then we have a topological/holomorphic theory.   In \cite{Kap06}, Kapustin observes that any $N=2$ theory admits a supercharge with this property. Thus, any $N=2$ theory admits a holomorphic/topological twist.  (In his paper, Kapustin focuses on theories which have finite $\beta$-function, because he considers theories which are topological in the stronger sense that they can be defined on oriented topological surfaces and not just framed topological surfaces.) 

Another construction of holomorphic/topological theories, which is more relevant to this paper, arises from $N=1$ gauge theory.  As shown in \cite{Cos13}, the $N=1$ pure gauge theory admits a deformation such that the deformed theory is still acted on by one supercharge $Q \in \mc{S}_+$, and that if we twisted this deformed theory using this supercharge, we find a theory which is holomorphic/topological.  We will discuss this deformation in detail later. 

\subsection{}
Now, suppose we have a holomorphic/topological theory, which is defined on four-manifolds of the form $\Sigma_{hol} \times \Sigma_{top}$, where the $\Sigma_{hol}$ is a Riemann surfaces and $\Sigma_{top}$ is a smooth surface. Depending on the details of the theory, there may be some restrictions on the topology of these surfaces, but all holomorphic/topological theories can be defined on such four-manifolds where $\Sigma_{top}$ is framed and $\Sigma_{hol}$ is equipped with a nowhere-vanishing holomorphic $1$-form. 

Let us suppose that this theory is equipped with two line operators in the topological direction.

The general yoga of topological field theory (\cite{Lur09}, \cite{Cos07a}) tells us that we should assign to the Riemann surface $\Sigma_{hol}$ a category $\mc{C}(\Sigma_{hol})$. 

We will often assume that $\mc{C}(\Sigma_{hol})$ is the category of vector spaces.  Without this assumption, the integrable system we construct will be a kind of generalized integrable system.
\begin{figure}
\includegraphics[scale=1.2]{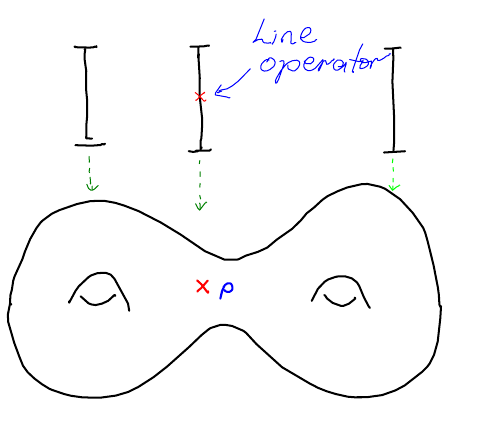}
\caption{The endpoint of a line operator is located at the centre of thee interval above $p$. Intervals above other points contain no line operators. This configuration yields the functor $\mc{F}_p$.\label{fig:line operator}}
\end{figure}
For a point $p \in \Sigma_{hol}$, we can put the end-point of a line operator at the point $(p,\tfrac{1}{2})$ in $\Sigma_{hol} \times [0,1]$, as in figure \ref{fig:line operator}.  This leads to a functor
$$
\mc{F}_p : \mc{C}(\Sigma_{hol}) \to \mc{C}(\Sigma_{hol}).
$$
In the case that $\mc{C}(\Sigma_{hol})$ is $\op{Vect}$, the functor $\mc{F}_p$ is given by tensor product with a vector space $V_p$. 

Suppose that $q \in \Sigma_{hol}$ is a point with $q \neq p$. If we place a different line operator in the fibre above $q$, we get a functor $\mc{G}_q : \mc{C}(\Sigma_{hol}) \to \mc{C}(\Sigma_{hol})$. In the case that $\mc{C}(\Sigma_{hol})$ is $\op{Vect}$, this functor is given by tensoring with a vector space $W_q$. 

Consider $\Sigma_{hol} \times I \times I$.  Placing one line operator in the interval $p \times \tfrac{1}{2} \times I$, and another on the interval $q \times I \times \tfrac{1}{2}$, as in figure \ref{fig:square}, leads to a natural isomorphism
$$
\mc{F}_p \circ \mc{G}_q \iso \mc{G}_q \circ \mc{F}_p.
$$
\begin{figure}
\includegraphics{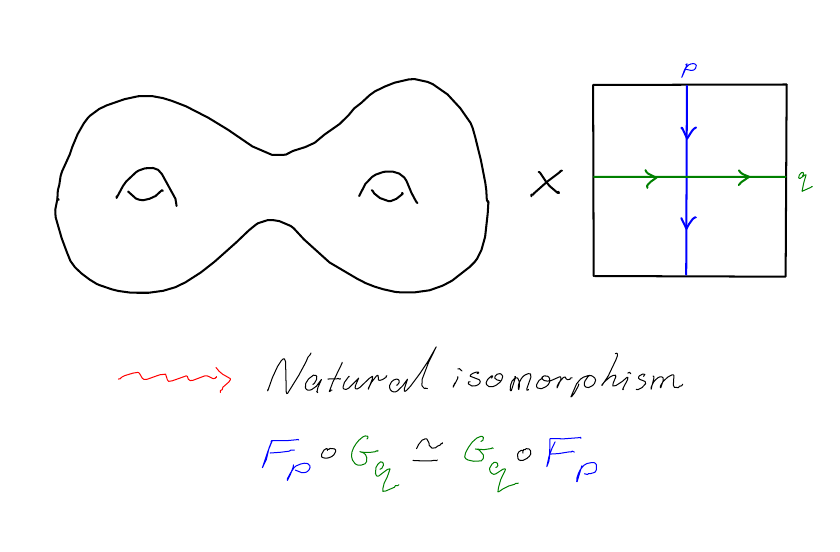}
\caption{The blue vertical line is located at $p \in \Sigma_{hol}$, and is labelled by one line operator. The green horizontal line is located at $q$ and labelled by another line operator. This configuration yields the natural isomorphism shown.
\label{fig:square}}
\end{figure}
In the case that $\mc{C}(\Sigma_{hol})$ is $\op{Vect}$, it leads to an isomorphism
$$
\check{R}(p,q): V_p \otimes W_q \iso W_q \otimes V_p.
$$
This isomorphism will be the $R$-matrix (or Boltzmann weights) of the integrable lattice model. 

\subsection{The Hilbert space and the transfer matrix}
Consider our theory on $\Sigma_{hol} \times S^1$.  Let us place the end points of $n$ line operators at the points $p \times 2 \pi k / n$, where $k$ ranges from $1$ to $n$, as in figure \ref{fig:hilbert space}.  The rules of topological field theory tell us that the Hilbert space of the theory, in the presence of these line operators, is the Hochschild homology group
$$
HH_\ast ( \mc{C}(\Sigma_{hol}, \mc{F}_p \circ \dots \circ \mc{F}_p))
$$
of the category $\mc{C}(\Sigma_{hol})$ with coefficients with the composition of $n$ copies of the functor $\mc{F}_p$.

\begin{figure}
\includegraphics{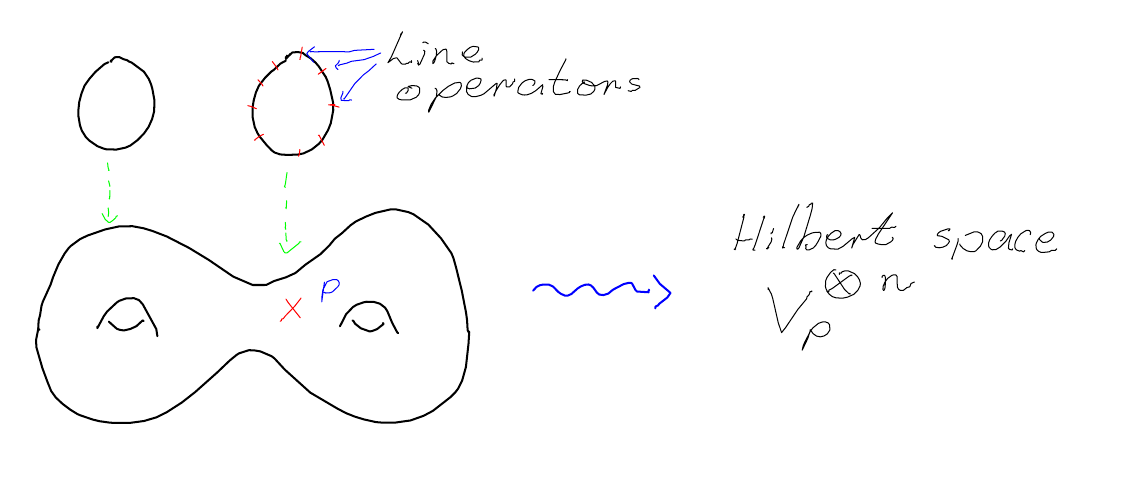}
\caption{Endpoints of $n$ line operators are distributed on the circle above $p$. The Hilbert space in the presence of these line operators is $V_p^{\otimes n}$ (in the case that the category  $\mc{C}(\Sigma_{hol})$ is vector spaces). \label{fig:hilbert space}}
\end{figure}
In the case that $\mc{C}(\Sigma_{hol})$ is the category of vector spaces, so that the functor $\mc{F}_p$ is tensor product with a vector space $V_p$, this formula tells us that the Hilbert space is just $V_p^{\otimes n}$, which is the Hilbert space of the lattice model.

Next let us consider the transfer matrix.  Consider the $4$-manifold with boundary $\Sigma_{hol} \times S^1 \times I$, equipped with the following embedded $1$-manifolds, as illustrated in figure \ref{fig:transfer matrix}:
\begin{enumerate}
\item There are $n$ vertical lines $p \times 2 \pi \tfrac{k}{n} \times I$ for $k = 1,\dots, n$.
\item There is one horizontal circle $q \times S^1 \times \tfrac{1}{2}$. 
\end{enumerate}
Let us label the $n$ lines at $p$ by one line operator, and the single circle at $q$ by another line operator, as in figure \ref{fig:transfer matrix}.
\begin{figure}
\includegraphics[scale=1]{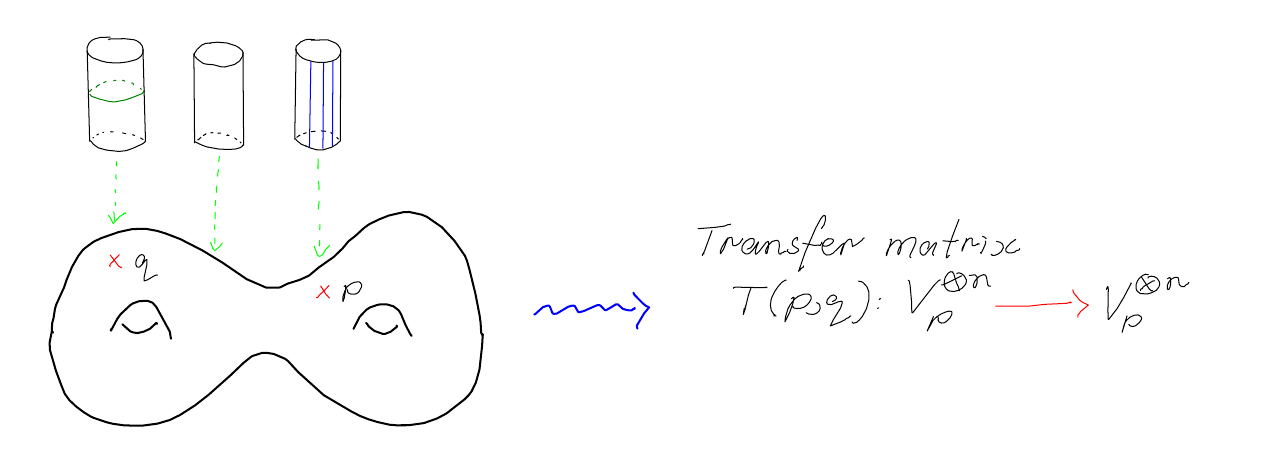}
\caption{There's a horizontal circle (green) at $q$ and $n$ vertical lines (blue) at $p$. This configuration yields the transfer matrix  $T(p,q)$. \label{fig:transfer matrix}}
\end{figure}

  This configuration gives rise to a linear operator on the Hilbert space in the presence of $n$ line operators. In the case that $\mc{C}(\Sigma_{hol})$ is the category of vector spaces, we therefore find a linear map
$$
T(p,q) : V_p^{\otimes n} \to V_p^{\otimes n}.
$$
\begin{lemma}
The linear map $T(p,q)$ is the transfer matrix, obtained as taking the trace in $W_q$ of the $n$-fold $W$-composition of 
$$\check{R}(p,q) : V_p \otimes W_q \to W_q \otimes V_p.$$
\end{lemma}
\begin{proof}
This is immediate from the axioms of topological field theory: we simply cut the cylinder into $n$ copies of the square $I \times I$ where each square contains one vertical and one horizontal line.  Each square contributes the $R$-matrix, and gluing the squares together corresponds to composition; gluing the resulting rectangle into a cylinder amounts to taking the trace. Thus, we find the transfer matrix. 
\end{proof}

\subsection{The transfer matrices commute}
Next, we will show the lattice model we have constructed is integrable.
\begin{lemma}
For all $q,q'$, we have
$$
[T(p,q), T(p,q')] = 0.
$$
\end{lemma}
\begin{figure}
\subfigure[$T(p,q)T(p,q')$]{
\includegraphics[scale=0.4]{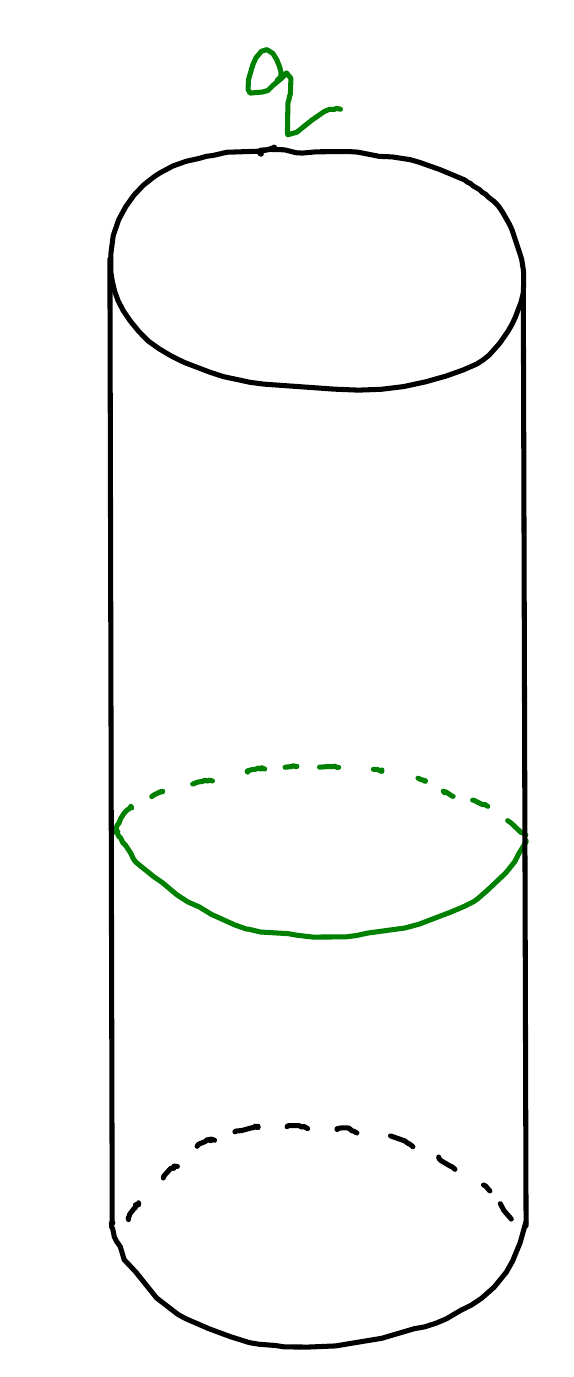} \includegraphics[scale=0.4]{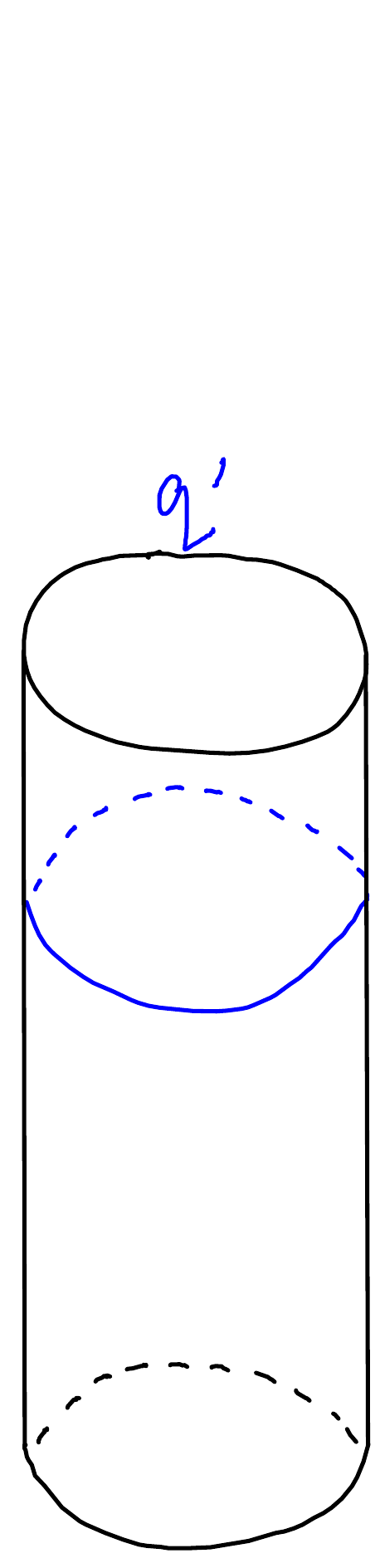}}
\hspace{40pt}
\subfigure[$T(p,q')T(p,q)$] {
\includegraphics[scale=0.4]{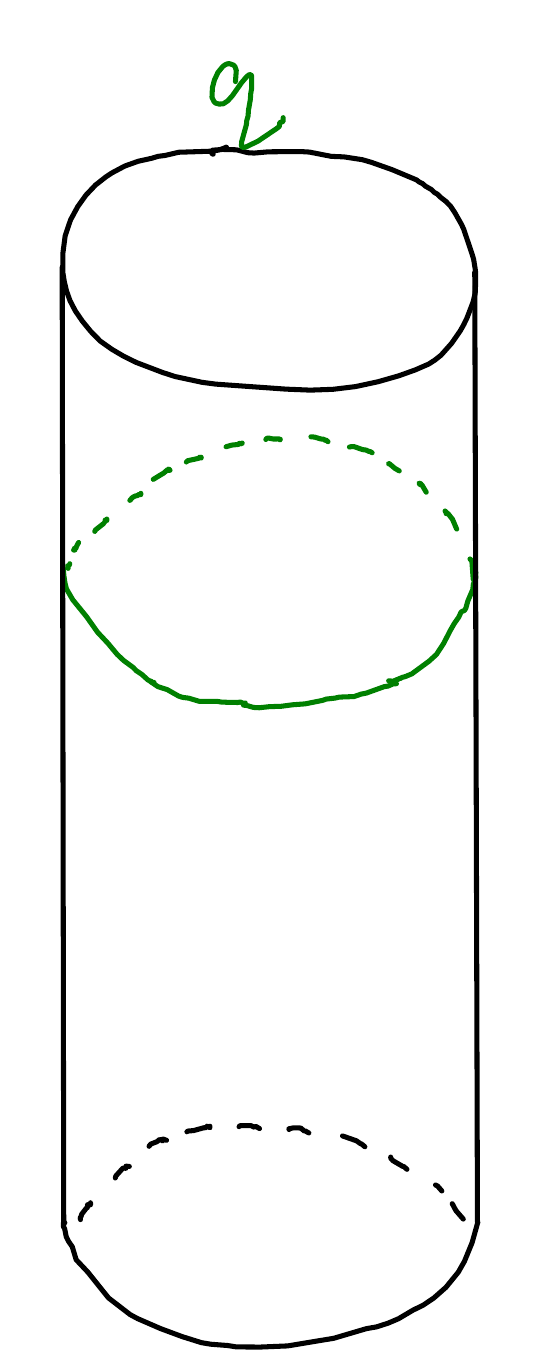} \includegraphics[scale=0.4]{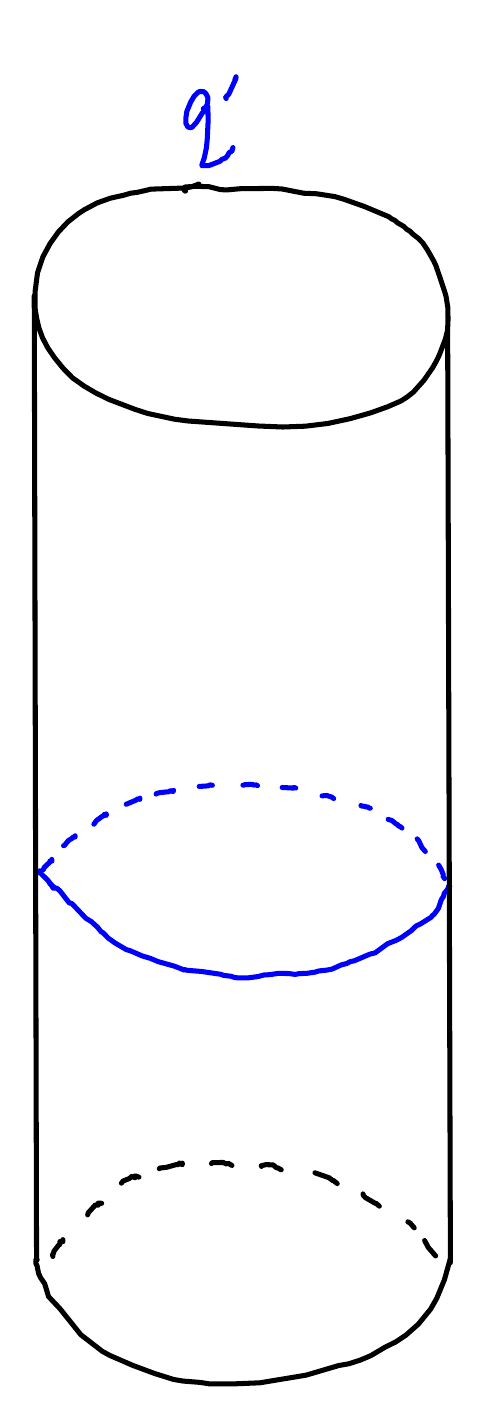}}
\caption{We get from (a) to (b) by sliding the blue circle up and the green circle down; this doesn't change the operator because our theory is topological in the directions containing the cylinder, and the two cylinders are located at different points in the $\Sigma_{hol}$ surface. \label{fig:commuting_transfer_matrices}}
\end{figure}
\begin{proof}
The proof is explained in figure \ref{fig:commuting_transfer_matrices}. The point is that $T(p,q) T(p,q')$ is obtained by placing the circle at $q$ above that at $q'$, whereas $T(p,q') T(p,q)$ is obtained by placing the circle at $q'$ above that at $q$. Because the theory is topological, we can slide the circles past each other. 
\end{proof}

In fact, in a similar way, we can show that the Yang-Baxter equation holds.  In order to state and prove the Yang-Baxter equation, we need to change notation a little bit. For any two line operators $U,U'$  in our theory, let 
$$
\check{R}(U_p, U'_q) : U_p \otimes U'_q \to U'_q \otimes U_p
$$
be the isomorphism arising from the construction described above. We have been using the notation $\check{R}(p,q)$ for this isomorphism in the case that $U = V$ and $U' = W$. 
\begin{figure}
\subfigure[The Yang-Baxter equation]{
\includegraphics[scale=0.7]{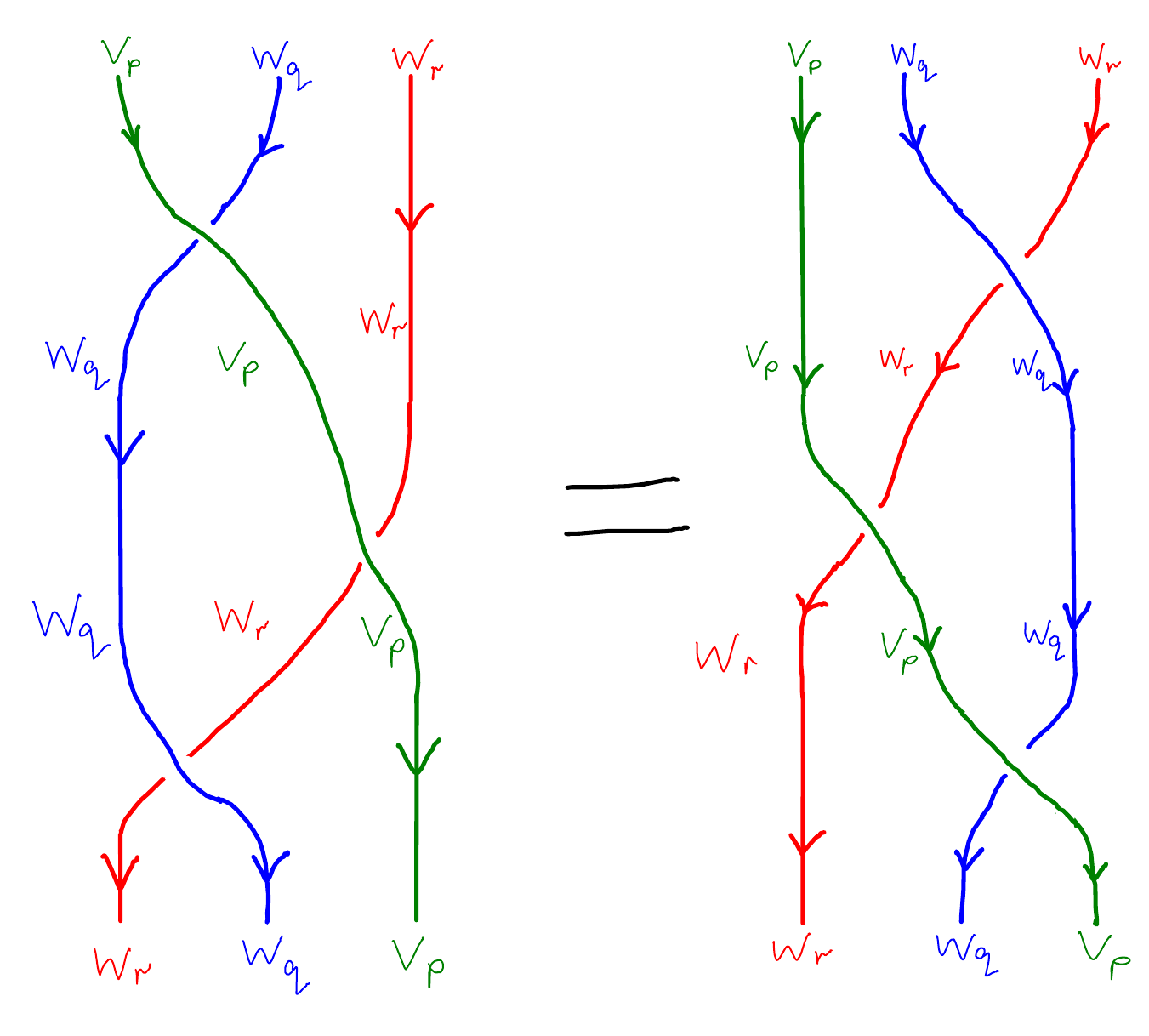}}
\subfigure[The relation $\check{R}(W_p, W_q) \check{R}(W_q,W_p) = 1$.]{\includegraphics[scale=0.5]{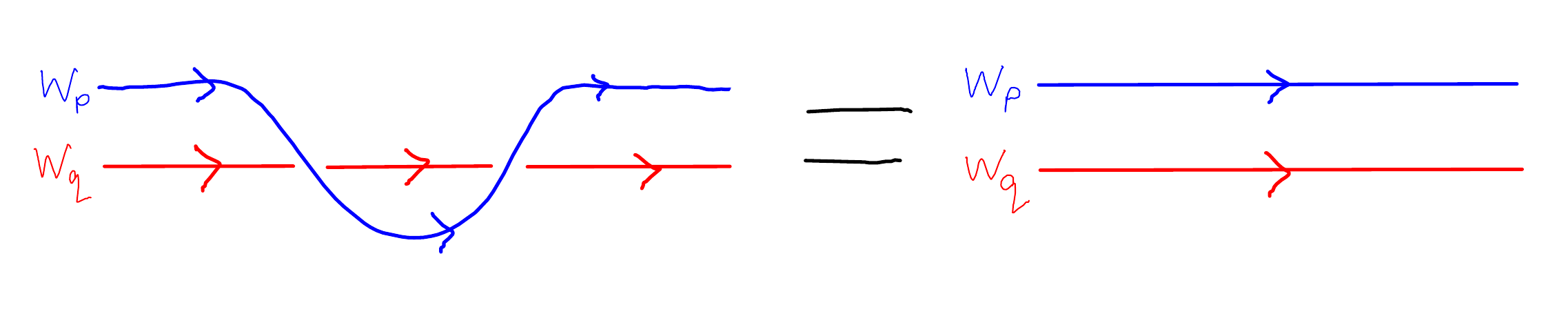}}
\caption[Yang-Baxter equation]{The Yang-Baxter equation, which is the braid relation with labelled strands; and the inverse relation. The diagrams should be interpreted as follows. The plane of the page is the topological plane of the theory. Each strand lives over a fixed point ($p$, $q$ or $r$) in the surface $\Sigma_{hol}$. The over and under crossings are dictated by the convention that $p$ is above $q$ and $q$ is above $r$ (i.e.\ we choose a path from $p$ to $r$ passing through $q$, and the ``height'' in the diagram corresponds to position on this path). \\ \parindent=10pt
As in figures \ref{fig:YBE} and \ref{fig:square}, a crossing corresponds to an $R$-matrix. For example, the $q$-$r$ crossing gives $\check{R} : W_q \otimes W_r \to W_r \otimes W_q$. The two sides are equal simply because the theory is topological in the plane in which the strands lie, so that we can slide the strands over each other without changing the operator. 
\label{fig:braid}}
\end{figure}
\begin{lemma}
For every triple of distinct points $p,q,r \in \Sigma_{hol}$, the Yang-Baxter equation 
$$
\check{R}(W_q, W_r)\check{R}(V_p, W_r)\check{R}({V}_p, {W}_q) = \check{R}({V}_p, {W}_q)\check{R}({V}_p, {W}_r)\check{R}({W}_q, {W}_r).
$$
holds\footnote{I'm very grateful to Josh Shadlen for discussions on the field-theoretic interpretation of the Yang-Baxter equation}. 

Similarly, for every pair of points $p,q$, the equation
$$
\check{R}(W_p, W_q) \check{R}(W_q,W_p) = 1
$$
holds.
\end{lemma}
\begin{proof}
The proof is illustrated in figure \ref{fig:braid}.\end{proof}
Thus, we have proved the following theorem.
\begin{theorem}
Suppose we have a four-dimensional field theory which is a mixture of topological and holomorphic.  Suppose that the theory is equipped with two line operators, and suppose that the category $\mc{C}(\Sigma_{hol})$ assigned to a fixed holomorphic surface is the category of vector spaces.  

Then, associated to this data is a two-dimensional integrable lattice model where the spectral parameter lives in an open subset of $\Sigma_{hol}$.
\end{theorem}

I should remark that this construction can be generalized in many ways.
\begin{enumerate}
\item We can include surface operators in the two topological directions as well as line operators.  Suppose we label some points $x_1,\dots,x_n \in \Sigma_{hol}$ by surface operators. Then, the same story holds, except that the line operators must live above the complement of the points $x_1,\dots,x_n$. 

The category $\mc{C}(\Sigma_{hol})$ is modified in the presence of surface operators.  Let us denote this modified category by $\mc{C}(\Sigma_{hol},\{x_1,\dots,x_n\})$. Even if $\mc{C}(\Sigma_{hol})$ is not the category of vector spaces, judicious choices of surface operators can ensure that $\mc{C}(\Sigma_{hol},\{x_1,\dots,x_n\})$ is the category of vector spaces, so that the construction above applies.
\item Suppose that we have an $n+2$ dimensional theory which is topological in $2$ directions and arbitrary (not necessarily holomorphic) in $n$ directions. Suppose that this theory is equipped with line operators in the topological direction.  Suppose that $M$ is a compact $n$-manifold such that the theory is defined on $M \times \R^2$, and suppose that the category associated to $M$ is the category of vector spaces. (This will happen if the space of classical solutions on $M \times \R^2$ is a point). Then, as above, we find a $2$-dimensional integrable lattice model.  The Boltzmann weights and transfer matrix of this lattice model depend smoothly on a pair of points in $M$. 

If we fix one point $p$ and a line operator above $p$, the Hilbert space will be $V_p^{\otimes n}$ where $V_p$ is the vector space arising from the line operator. The partition function of the theory is a smooth function on $M \setminus\{p\}$. 
\end{enumerate}

\subsection{}
Before proceeding to a discussion of our main example of this construction, I should remark on a possible inconsistency in the discussion.  I have explained how to construct integrable lattice models using the Atiyah-Segal-Freed style axioms for (topological) field theory, where the Hilbert space and its categorical analogs are the fundamental objects.  However, the main example will be constructed using the technology of \cite{Cos11,CosGwi11} which uses a different axiom system, that of factorization algebras.  There is an apparent discrepancy, because the factorization algebra encodes the operators (or observables) of a theory, but not always the Hilbert space.

This discrepancy is resolved by using a version of the state-operator correspondence. If we have a codimension $1$ manifold which is the boundary $\partial U$ of a region $U$, we declare that the Hilbert space for $\partial U$ is the space of observables on $U$.  

Similarly, we posit that a categorified version of the state-operator correspondence holds.  For example, in a four-dimensional topological/holomorphic theory, the Atiyah-Segal-Freed axioms tell us that a monoidal category should be assigned to the product of $\Sigma_{hol}$ with $\partial ([0,1])$, the boundary of the interval. The theory of factorization algebras tells us that observables on $\Sigma_{hol} \times [0,1]$ (or equivalently on $\Sigma_{hol} \times D$ where $D$ is a two-disc) forms an $E_2$ algebra, where the $E_2$ structure arises from the operator product in the two topological directions. A theorem of Lurie \cite{Lur12} tells us that the category of left modules for an $E_2$ algebra is a monoidal category. The categorified state-operator correspondence we propose tells us that this monoidal category is what we assign to $\Sigma_{hol} \times \partial([0,1])$.   

The reader can check that in our construction of integrable systems, we didn't really need the category $\mc{C}(\Sigma_{hol})$ associated to $\Sigma_{hol}$; we only used the monoidal category of endofunctors of this category associated to $\Sigma_{hol} \times \partial([0,1])$.  In this manner, we see that the arguments presented above which produce an integrable system using the Atiyah-Segal-Freed axioms for field theory also work using the language of factorization algebras.  The interested reader can consult \cite{Cos13} for more on this point, as \cite{Cos13} is written entirely using the language of factorization algebras.

\section{Spin-chains from $N=1$ pure gauge theory}
So far, we have explained the general construction which associates an integrable lattice model to a four dimensional field theory which is a mixture of holomorphic and topological.  In this section, I will sketch the main result of \cite{Cos13}: I will describe such a four-dimensional theory where the corresponding integrable lattice model is a spin-chain model.  A wide class of spin-chain models are constructed this way, by varying the gauge  group and the Wilson operator used: we find every spin-chain model that arises from the representation theory of the Yangian.  In particular, we find the Heisenberg $XXX$ model, which arises from the Yangian for $\mf{sl}_2$; but not the $XXZ$ model, which is related to the quantum affine algebra for $\mf{sl}_2$ at non-zero level. 

The theory we are considering is a deformation of pure $N=1$ gauge theory.  Let us start by describing the fields and action functional of $N=1$ gauge theory on $\R^4$.  Recall that $\op{Spin}(4) = \op{SU}(2) \times \op{SU}(2)$.  Let $\mc{S}_+$ and $\mc{S}_-$ be the defining two-complex-dimensional representations of the two copies of $\op{SU}(2)$. These are the irreducible spin representations of $\op{Spin}(4)$.  We use the notation
$$
\mscr{S}_{\pm} = \mc{S}_{\pm} \otimes \cinfty(\R^4)
$$
for sections of the corresponding spin bundles.

We will describe the $N=1$ gauge theory in the first-order formulation (which is equivalent to the more familiar second-order formulation). Let $\g$ be a semi-simple Lie algebra.  The fields of the $N=1$ gauge  theory consist of a connection $A \in \Omega^1(\R^4) \otimes \g$, and adjoint-valued self-dual two-form $B \in \Omega^2_+(\R^4)$, and adjoint-valued spinors $\psi_{\pm} \in \mc{S}_{\pm} \otimes \g$. The action functional is 
$$
S(A,\psi) = \int F(A) \wedge B + c \int B \wedge B + \int \psi_{+} \Dirac_A \psi_-.
$$
The gauge coupling constant is $c$. 
 
We will consider all of our fields to be complex, so that $A$, $B$ are a complex $1$-form and $2$-form. The action functional and all observables are holomorphic functions of the fields, and the path integral will be taken over a contour.  Because we work in perturbation theory the choice of contour is irrelevant.  We need to proceed in this way because the spin representations $\mc{S}_{\pm}$ are complex.  This is an artifact of working in Euclidean signature.  In Lorentzian signature,  the spin representation $\mc{S}_+ \oplus \mc{S}_-$ is real so that we can take our fields to be real. 

Let us consider a deformation of this action functional. Let us choose a complex structure on $\R^4$, and write $\R^4 = \C^2$ with coordinates $z$ and $w$.  Let 
$$
S'(A,\psi) = - \int \frac{1}{2}z F(A) \wedge F(A) + 2c \int \ip{\psi_+ ,\d z\psi_- }
$$
where in the second term, we are using the Clifford multiplication of $\d z$ on $\psi_+$ and then pairing with $\psi_-$. By integration by parts, we can rewrite the first term as $\int \d z CS(A)$ where $CS(A)$ is the Chern-Simons three-form, normalized so that $\d CS(A) = \tfrac{1}{2} \ip{F(A),F(A)}$.

The deformed action is 
$$
S_{deformed} = S + \tfrac{1}{2 \pi i} S'.
$$

The $N=1$ supersymmetric gauge theory is acted on by the $N=1$ supertranslation Lie algebra, which is a complex Lie algebra whose underlying $\Z/2$-graded vector space is
$$
\mc{T}^{N=1} = \pi \left( \mc{S}_+ \oplus \mc{S}_-\right) \oplus \C^4.
$$
The symbol $\pi$ indicates parity reversal, so that the spinors $\mc{S}_{\pm}$ are odd. 

There is a unique up to scale $\op{Spin}(4)$-equivariant isomorphism
$$
\Gamma : \mc{S}_+ \otimes \mc{S}_- \to \C^4.
$$
The Lie bracket on $\mc{T}^{N=1}$ is defined by saying that if $Q_{\pm} \in \mc{S}_{\pm}$, 
$$
[Q_+, Q_-] = \Gamma(Q_+ \otimes Q_-)
$$
and all other brackets are zero. 

Elements of $\mc{S}_+$ induce complex structures on $\R^4$. One way to see this is to observe that the stabiliser of an element $Q \in \mc{S}_+$ is $\op{SU}(2)$ inside  $\op{Spin}(4)$.  A more concrete way to show this is as follows. If $Q \in \mc{S}_+$, the image of bracketing with $Q$ is a rank-two subspace of $\C^4 = \R^4 \otimes \C$, which we declare to be the $-i$ eigenspace of an operator $J : \R^4 \to \R^4$ defining the complex structure.    

In this way one can identify the projective space $\mbb{P}(\mc{S}_+)$ with the space of complex structures on $\R^4$ for which the standard Euclidean metric is K\"ahler and for which the induced orientation is the standard orientation on $\R^4$.  Similarly, $\mbb{P}(\mc{S}_-)$ is the space of complex structures compatible with the metric and inducing the opposite orientation. 

In \cite{Cos13} I show the following, by explicit computation. 
\begin{lemma}
Let $Q$ be the unique up to scale supercharge in $\mc{S}_+$ such that $z$ is holomorphic in the corresponding complex structure. (This happens if $Q$ is in the kernel of the Clifford multiplication map $\d z \cdot : \mc{S}_+ \to \mc{S}_-$). 

Then, the deformed $N=1$ supersymmetric gauge theory introduced above is invariant under $Q$.
\end{lemma}
\begin{remark}
A closely related result is proved in \cite{GaiWit08}, where they consider a similar deformation of the $N=4$ gauge theory. In fact Gaiotto and Witten show that the deformation they consider is $\tfrac{1}{2}$-BPS.  Probably the deformation considered here is also $\tfrac{1}{2}$-BPS, i.e. it is probably also invariant under the supercharge in $\mc{S}_-$ for which $z$ is holomorphic.  
\end{remark}
\begin{remark}
The action of supersymmetry is deformed when we deform the action functional.  The chosen supercharge $Q \in \mc{S}_+$ acts as follows.  Once we have chosen $Q$, we can identify $\mscr{S}_-$ with $\Omega^{1,0}(\C^2)$ for the chosen complex structure. Then, the deformed action of $Q$ on the fields of the theory has a term mapping
\begin{align*}
\mscr{S}_- = \Omega^{1,0}(\C^2) &\to \Omega^2_+(\C^2) \\
\psi_- & \mapsto \d z \wedge \psi_-. 
\end{align*}
\end{remark}
In what follows, we will fix this supercharge $Q \in \mc{S}_+$ and use the induced complex structure to identify $\R^4$ with $\C^2$. We will let $z,w$ be holomorphic linear coordinates on $\C^2$. 
\subsection{}
Since $S_{deformed}$ is invariant under the supercharge $Q$, we can consider the twisted theory. 
\begin{lemma}
The fields of the twisted theory are $A \in \Omega^1(\C^2) / \d z$ with action functional $\int \d z CS(A)$.
\end{lemma}
\begin{remark}
The observables of the twisted theory are, by definition, the $Q$-cohomology of the original observables.  The fields of the twisted theory are the $Q$-cohomology of the original fields. We set things up so that $Q$ acts linearly on the space of fields. When we take $Q$-cohomology, two components of $B \in \Omega^2_+$ cancel with the spinors in $\mscr{S}_+$. The remaining component of $B$ cancels with one of the components of the spinor in $\mscr{S}_-$. The remaining component of $\mscr{S}_-$ cancels with the $\d z$ component of the connection $A \in \Omega^1$, leaving a $3$-component partial connection.   
\end{remark}
\begin{remark}
As before, we treat the space of fields as a complex manifold, and integrate over a contour. 
\end{remark}
\begin{lemma}
The twisted, deformed $N=1$ theory is holomorphic in the $z$-plane and topological in the $w$-plane.  
\end{lemma}
This is easy to see from the presentation of the twisted theory give above with fields in $A \in \Omega^1 / \d z \otimes \g$ and action functional $\int \d z CS(A)$.

One manifestation of this is the following.
\begin{lemma}
A solution to the equations of motion to the twisted, deformed $N=1$ theory is a holomorphic bundle on $\C^2$ equipped with a holomorphic (and therefore flat) connection in the $w$-plane.
\end{lemma}
\subsection{}
We have described our twisted, deformed $N=1$ theory at the classical level.  One can ask when if it can be quantized, i.e.\ if there are no quantum anomalies. (I use the term quantization in the sense of \cite{Cos11b, CosGwi11}: in particular, we are working in perturbation theory).

It turns out that it can be quantized essentially uniquely on a range of backgrounds.   Let $X$ be a complex surface equipped with a non-zero closed holomorphic $1$-form $\alpha$ (which plays the role of $\d z$). Then, we can define the twisted, deformed $N=1$ gauge theory on $X$. The fields are $\Omega^1(X) / \alpha \otimes \g$, and the action functional, as above, is
$$
S(A) = \int \alpha \wedge CS(A).
$$
\begin{proposition}
Suppose that $X$ is equipped with a holomorphic volume form.  Then, the twisted deformed $N=1$ theory on $X$ admits an essentially unique quantization.
\end{proposition}
This is proved in \cite{Cos13} by analysing the obstruction-deformation group controlling quantizations.  More precisely, I calculate the group which controls quantizations with certain additional properties: the quantization must behave well in the limit when $\alpha$ becomes zero, in which case the theory is holomorphic BF theory; and the quantization must be compatible with symmetries of $X$ preserving the holomorphic volume form.  I show that the group containing possible anomalies vanishes, as does the group containing possible deformations (compatible with these symmetries).  It follows that there are no anomalies and the quantum theory is unique. 

One fact which which makes this theory reasonably tractable is that in a certain gauge, it is one-loop exact.  The required gauge is where $\dbar^\ast A = 0$.  There are other natural gauges: we could require that $\d^\ast_{\Sigma_{top}} A = 0$ and that $\dbar^\ast_{\Sigma_{hol}} A = 0$.  In this gauge, the theory is not one-loop exact. 

The theory can be quantized on more general backgrounds.  One generalization, which we need later, is to consider complex surfaces $X$ equipped with a complex curve $D \subset X$ and a trivialisation of $K_X(2 D)$. The closed holomorphic $1$-form $\alpha$ used to define the action functional can have a second-order pole along $D$. The one-form $\alpha$ determines a holomorphic vector field by $V \vee \omega = \alpha$, where $\omega$ is the meromorphic volume form on $X$.  In this situation, we modifying the theory by requiring that the fields and gauge transformations are trivial on $D$.  Thus, the solutions to the equations of motion are holomorphic bundles on $X$, trivialised on $D$, together with a holomorphic connection in the direction given by the holomorphic vector field $V$. 

The only example of this more general version of the theory we will consider is when $X$ is the product of $\mbb{P}^1$ and an elliptic curve $E$, and the divisor $D$ is $\infty \times E$. Let $z$ be a coordinate on $\mbb{P}^1$ and $w$ be a coordinate on $E$. The meromorphic volume form is $\d z \d w$, the meromorphic one-form is $\d z$, and the holomorphic vector field $ V$ is $\partial_w$.  The solutions to the equations of motion in this case are holomorphic $G$-bundles on $\mbb{P}^1 \times E$ together with a holomorphic connection in the $E$-direction, where both bundle and connection are trivialised on $\infty \times E$.

\subsection{}
This deformed version of $N=1$ gauge theory has an invariant Wilson operator.  Recall that we wrote the theory in the first-order formalism, with an auxiliary field $B \in \Omega^2_+ \otimes \g$.  Once we have chosen a complex structure on $\R^4$, we can identify $\Omega^2_+$ as 
$$
\Omega^2_+ = \Omega^{2,0} \oplus \Omega^0 \cdot \omega \oplus \Omega^{0,2}
$$ 
where $\omega \in \Omega^{1,1}$ is the K\"ahler form coming from the metric on $\R^4$.    

\begin{lemma}
The connection in the $w$-plane defined by $A + B^{2,0} \d w$ is invariant under the supercharge $Q \in \mc{S}_+$. 
\end{lemma}
Again, this follows by a simple explicit computation.  

In this way, we can construct a classical $Q$-invariant Wilson operator in our theory, for any line in the $w$-plane and for every representation $V$ of $\g$. 
\begin{theorem}
Suppose that $V$ lifts to a representation of the Yangian $Y(\g)$. Then, the Wilson operator associated to $V$ lifts to a quantum Wilson operator, defined on straight lines in the $w$-plane. 
\end{theorem}
This result is rather subtle, and relies on the main abstract result of \cite{Cos13}, which says that the operator product of the twisted, deformed $N=1$ gauge theory in the $w$-plane is controlled by the Yangian. In the case that $\g = \sl_n$, every representation lifts to a representation of the Yangian.

\section{Spin chains and $N=1$ gauge theory}
So far, I have described a general construction of integrable lattice models from holomorphic/topological four-dimensional theories equipped with line operators.  I have also explained how to construct such a theory together with a line operator from a deformation of $N=1$ supersymmetric gauge theory. The main result is that the integrable system which arises in this way is a certain spin-chain system. 

If $\g$ is a semi-simple Lie algebra, Drinfeld has constructed a Hopf algebra $Y(\g)$ which quantizes the universal enveloping algebra of $\g[[t]]$.  Further, he shows how to construct, from every finite-dimensional representation $V$ of $Y(\g)$, an integrable lattice model.  The $R$-matrix encoding the Boltzmann weights arises from the universal $R$-matrix 
$$
R(z) \in Y(\g) \otimes Y(\g) ((z))
$$
by applying the homomorphism from $Y(\g) \otimes Y(\g)$ to $\op{End}(V \otimes V)$.   We will call this integrable system the spin-chain model associated to $\g$ and $V$.

The simplest case of this construction is when $\g = \sl_2$ and $V$ is the fundamental representation, in which case we find the Heisenberg $XXX$ model. 

In the four-dimensional holomorphic/topological theory arising from $N=1$ gauge theory with gauge Lie algebra $\g$, there is a classical Wilson operator associated to every representation $V$ of $\g$. I have also mentioned that a lift of this to a quantum Wilson operator is the same as a lift of the representation of $\g[[t]]$ to a representation of the Yangian $Y(\g)$, which quantizes $U(\g[[t]])$. 

\begin{theorem}
The integrable lattice model arising from the twisted, deformed  $N=1$ gauge theory, with a Wilson operator coming from the representation $V$ of the Yangian $Y(\g)$, is spin-chain system constructed by Drinfeld from $\g$ and $V$. 
\end{theorem}
As a corollary, we can compute the expectation value of Wilson operators in the twisted, deformed $N=1$ gauge theory.   Let $E$ be an elliptic curve, and consider the theory on $\C_z \times E_w$; so that we compactify the topological $w$-direction to the elliptic curve $E$.   For a point $z \in \C$, a representation $V$ of the Yangian, and an $a$- or $b$- cycle on the elliptic curve $E$, let $\chi_a(z,V)$ (respectively, $\chi_b(z,V)$) denote the Wilson operator in the representation $V$ placed on the circle $z \times a$ (respectively, $z \times b)$. Since the theory is topological on $E$, the Wilson operator only depends on the homotopy class of the circle in $E$. 

Then, we have the following.
\begin{corollary}
Let $a_1,\dots,a_m$ and $b_1,\dots,b_m$ be disjoint $a$- and $b$-cycles. Then, the vacuum expectation value 
$$
\ip{\chi_{a_1}(0,V), \dots ,\chi_{a_n} (0,V) ,\chi_{b_1}(z,V), \dots, \cdots_{b_m}(z,V)}
$$
is the partition function of the spin-chain integrable lattice model associated to $\g$ and $V$ on the $n \times m$ doubly periodic lattice, with spectral parameter $z$. 
\end{corollary}
\begin{remark}
\begin{enumerate}
\item I should explain what I mean by vacuum expectation value. As I remarked earlier, the theory is defined on $\mbb{P}^1_z \times E_w$, where all fields and gauge transformations are trivial on $\infty \times E_w$.  It is easy to check that there are no massless modes when we perturb around the trivial solution to the equations of motion (i.e. the trivial solution can not be deformed).  In fact, the trivial solution is the solution that is stable (in the sense of the theory of moduli of bundles).   It follows that we can define vacuum expectation values in perturbation theory. 
\item   If, instead of requiring that the bundle and connection are trivial at infinity, we ask that the monodromy around the $a$- and $b$-cycles at infinity are given by some fixed commuting elements of $G$, we find the partition function of the integrable lattice model on an $n\times m$ lattice with twisted boundary conditions. 
\item Another variation is to place the Wilson operators on $a$-cycles at independent points $z_1,\dots,z_n$ in $\C$, and the Wilson operators on $b$-cycles at independent points $z_1', \dots,z_m'$ in $\C$. We must have $z_i \neq z_j'$ for all $i$ and $j$.  This corresponds to introducing ``inhomogeneities'' in the integrable lattice model.
\end{enumerate}
\end{remark}

\subsection{}
Most of \cite{Cos13} is devoted to the proof of this correspondence between spin-chain systems and $N=1$ gauge theory, so I will only sketch the main idea here rather than giving details.  I will sketch the proof using more physical language than in \cite{Cos13}.

There is only one stable solution to the equations of motion of this field theory on $\mbb{P}^1_z \times \C_w$.  It follows that the category $\mc{C}(\mbb{P}^1_z)$ is just $\op{Vect}$, the category of vector spaces.

Therefore, the formal picture described earlier tells us that this field theory leads to some integrable system. A related formal argument, based on Koszul duality, tells us that underlying this integrable system is a Hopf algebra equipped with an $R$-matrix satisfying the Yang-Baxter equation. (I will say a little more about this below).  We can calculate the classically this Hopf algebra is $U(\g[[t]])$, so whatever Hopf algebra we find must be a quantization of this.   Drinfeld shows that the Yangian is the unique quantization of the Hopf algebra $U(\g[[t]])$ into a Hopf algebra with an $R$-matrix of the form
\begin{align*}
R &= 1 + \hbar \frac{c}{t_1 + z - t_2} + O(\hbar^2) \\
&= 1 + \hbar c \sum_{i \ge 0} z^{-i-1}(t_2 - t_1)^i + O(\hbar^2) \in U(\g[[t_1]]) \otimes U(\g[[t_2]]) ((z))[[\hbar]].
\end{align*}
Here $c \in \g \otimes \g$ is the quadratic Casimir.

The leading contribution to the $R$-matrix can be calculated explicitly in our theory, by calculating the leading contribution to the operator product of Wilson lines. This implies that the Hopf algebra must be the Yangian, and that the $R$-matrix must be the one constructed by Drinfeld. 

\subsection{} 
I should clarify why one finds a Hopf algebra, and why the classical limit of this Hopf algebra is $U(\g[[t]])$.  The operator product in the $w$-plane gives the space of local observables (i.e. supported on a point) the structure of an $E_2$ algebra.  A result of Tamarkin \cite{Tam07} shows that any augmented $E_2$ algebra can be turned into a Hopf algebra, by a procedure called Koszul duality. See \cite{CosSch13} for a discussion of this construction. 

One can easily calculate the classical local observables of the theory. Classical observables are functions on the equations of motion, which is the moduli of holomorphic bundles with a flat connection in the $w$-direction.  Locally, all such bundles are trivial. A naive analysis would then suggest that there are no local observables.  However, one has to take account of ghosts, arising from automorphisms of the trivial bundle.

On $D_z \times D_w$, the Lie algebra of infinitesimal automorphisms of the trivial such bundle is $\g \otimes \op{Hol}(D_z)$.  Replacing $D_z$ by a formal disc, we find $g[[z]]$. The space $\g[[z]]$ is therefore the space of ghosts that enter into local observables. Since ghosts are fermionic, we find that local observables are the exterior algebra on the dual of $\g[[z]]$.  This exterior algebra has a BRST differential, which can be identified with the Chevalley-Eilenberg differential.  Therefore, classical local observables are the Chevalley-Eilenberg cochain complex $C^\ast(\g[[z]])$. 

It is a standard result that the Koszul dual of the Chevalley-Eilenberg cochain complex of a Lie algebra is the universal enveloping algebra. It follows that the Hopf algebra arising from classical observables is $U(\g[[z]])$.  The Hopf algebra arising from quantum observables is a deformation of this. By Drinfeld's uniqueness result, together with a one-loop calculation, we find that the quantum Hopf algebra is the Yangian, as desired.  

\section{Integrable systems from $N=2$ theories}
Kapustin  observed that any $N=2$ field theory has enough supersymmetry to have a holomorphic/topological twist.  The constructions of this paper therefore imply that one can construct integrable systems from such a theory which encode the behaviour of line operators in the topological directions.  

Let $G$ be a compact Lie group acting on a hyperk\"ahler manifold $M$. To this data one can associate a classical $N=2$ theory, the four-dimensional gauged $\sigma$-model.  Before twisting, this theory is not renormalizable. However, one can show (along the lines of the analogous result in \cite{Cos13}) that the holomorphic/topological twist of this theory admits a unique quantization compatible with certain symmetries.  The twisted theory only depends on $M$ as a holomorphic symplectic manifold. In this way, one gets a very general construction of integrable systems and solutions to the Yang-Baxter equation. It is a very interesting problem to explicitly calculate the solutions to the YBE arising from particular Wilson and t'Hooft operators in these theories. The $R$-matrix is encoded in the OPE between the line operators. A related OPE was calculated by Moraru and Saulina in \cite{MorSau12}; the calculated the OPE in the topological direction of parallel Wilson and t'Hooft operators.


\begin{thebibliography}{Cos11b}

\bibitem[CG12]{CosGwi11}
K.~Costello and O.~Gwilliam, \textsl{ Factorization algebras in perturbative
  quantum field theory},
\newblock Available at \\
  \verb=http://math.northwestern.edu/~costello/renormalization.html=  (2012).

\bibitem[Cos07]{Cos07a}
K.~Costello, \textsl{ Topological conformal field theories and {C}alabi-{Y}au
  categories},
\newblock Adv. Math. \textbf{ 210}(1), 165--214 (2007).

\bibitem[Cos11a]{Cos11b}
K.~Costello, \textsl{ Notes on supersymmetric and holomorphic field theories in
  dimensions $2$ and $4$},
\newblock (2011), {arXiv:1111.4234}.

\bibitem[Cos11b]{Cos11}
K.~Costello,
\newblock \textsl{ Renormalization and effective field theory},
\newblock Surveys and monographs, American Mathematical Society, 2011.

\bibitem[Cos13]{Cos13}
K.~Costello, \textsl{ Supersymmetric gauge theory and the {Y}angian},
\newblock (2013), {arXiv:1303.2632}.

\bibitem[CS13]{CosSch13}
K.~Costello and C.~Scheimbauer,
\newblock Lectures on Mathematical aspects of (twisted) supersymmetric gauge
  theories,
\newblock 2013.

\bibitem[DF99]{DelFre99}
P.~Deligne and D.~Freed,
\newblock \textsl{ Quantum fields and strings: a course for mathematicians,
  {V}ol. 1, 2 ({P}rinceton, {NJ}, 1996/1997)}, chapter Supersolutions,
\newblock Amer. Math. Soc., Providence, RI, 1999.

\bibitem[Fer11]{Fer11}
L.~Ferro, \textsl{ Yangian Symmetry in N=4 super Yang-Mills},
\newblock (2011), {arXiv:1107.1776}.

\bibitem[Fre99]{Fre99}
D.~S. Freed,
\newblock \textsl{ Five Lectures on Supersymmetry},
\newblock American Mathematical Soc., 1999.

\bibitem[GW08]{GaiWit08}
D.~Gaiotto and E.~Witten, \textsl{ Janus Configurations, {C}hern-{S}imons
  Couplings, And The Theta-Angle in N=4 Super {Y}ang-{M}ills theory},
\newblock (2008), {arXiv:0804.2907}.

\bibitem[Kap06]{Kap06}
A.~Kapustin, \textsl{ Holomorphic reduction of N=2 gauge theories, Wilson-'t
  Hooft operators, and S-duality},
\newblock (2006), {arXiv:hep-th/0612119}.

\bibitem[Lur09]{Lur09}
J.~Lurie, \textsl{ On the classification of topological field theories},
\newblock (2009), {arXiv:0905.0465}.

\bibitem[Lur12]{Lur12}
J.~Lurie, \textsl{ Higher algebra},
\newblock (2012).

\bibitem[MO12]{MauOko12}
D.~Maulik and A.~Okounkov, \textsl{ Quantum cohomology and quantum groups},
\newblock (2012), {arxiv:1211.1287}.

\bibitem[MS12]{MorSau12}
R.~Moraru and N.~Saulina, \textsl{ {OPE} of {W}ilson-'t {H}ooft operators in
  {N=4} and {N=2 SYM} with gauge group {G=PSU(3)}},
\newblock (2012), {arXiv:1206.6896}.

\bibitem[NS09]{NekSha09}
N.~Nekrasov and S.~Shatashvili, \textsl{ Quantization of Integrable Systems and
  Four Dimensional Gauge Theories},
\newblock (2009), {arXiv:0908.4052}.

\bibitem[Tam07]{Tam07}
D.~Tamarkin, \textsl{ Quantization of {L}ie bialgebras via the formality of the
  operad of little disks},
\newblock Geom. Funct. Anal. \textbf{ 17}(2), 537--604 (2007).

\bibitem[Yam13]{Yam13}
M.~Yamazaki,
\newblock New Integrable Models from the Gauge/YBE Correspondence, 2013.

\end{thebibliography}
\def\cprime{$'$}

\end{document}